\documentclass[11pt,letter]{article}
\usepackage{authblk}
\usepackage[margin=1in]{geometry}
\usepackage[utf8x]{inputenc}
\usepackage{amsthm}
\usepackage{wrapfig,floatflt,graphicx,amssymb,textcomp,array,amsmath}
\usepackage{enumerate,enumitem}
\usepackage{multirow}
\usepackage{tabularx}
\usepackage{color}
\usepackage{todonotes}
\usepackage[titletoc,title]{appendix}

\usepackage{lineno}

\newcommand{\etal}{{et~al.}}
\newcommand{\len}[1]{\mathrm{len}(#1)}

\title{Improved approximation ratios for two Euclidean maximum spanning tree problems
}

\author{Ahmad Biniaz\thanks{This research is supported by NSERC.}
}

\affil{School of Computer Science, University of Windsor\\\texttt{ahmad.biniaz@gmail.com}}

\date{}
\newtheorem{lemma}{Lemma}

\newtheorem{theorem}{Theorem}
\newtheorem{observation}{Observation}
\newtheorem*{problem*}{Problem}
\newtheorem*{claim*}{Claim}
\newtheorem*{invariant*}{Invariant}

\renewcommand\footnotemark{}
\begin{document}
	\maketitle
	\vspace{-10pt}
	\begin{abstract}
		We study the following two maximization problems related to spanning trees in the Euclidean plane. It is not known whether or not these problems are NP-hard. We present approximation algorithms with better approximation ratios for both problems. The improved ratios are obtained mainly by employing the Steiner ratio, which has not been used in this context earlier.
		 
		(i) {\em Longest noncrossing spanning tree}: Given a set of points in the plane, the goal is to find a maximum-length noncrossing spanning tree. Alon, Rajagopalan, and Suri (SoCG 1993) studied this problem for the first time and gave a $0.5$-approximation algorithm. Over the years, the approximation ratio has been successively improved to $0.502$, $0.503$, and to $0.512$ which is the current best ratio, due to Cabello \etal. We revisit this problem and improve the ratio further to $0.519$. The improvement is achieved by a collection of ideas, some from previous works and some new ideas (including the use of the Steiner ratio), along with a more refined analysis. 	 
		 
		(ii) {\em Longest spanning tree with neighborhoods}: Given a collection of regions (neighborhoods) in the plane, the goal is to select a point in each neighborhood so that the longest spanning tree on selected points has maximum length. We present an algorithm with approximation ratio $0.524$ for this problem. The previous best ratio, due to Chen and Dumitrescu, is $0.511$ which is in turn the first improvement beyond the trivial ratio $0.5$. Our algorithm is fairly simple, its analysis is relatively short, and it takes linear time after computing a diametral pair of points. The simplicity comes from the fact that our solution belongs to a set containing three stars and one double-star. The shortness and the improvement come from the use of the Steiner ratio.
		
	\end{abstract}
\section{Introduction}

Spanning tree is a well-studied and fundamental structure in graph theory and combinations. The well-known minimum spanning tree (Min-ST) problem asks for a spanning tree of a weighted graph, with minimum total edge weight. In contrast, the maximum spanning tree (Max-ST) problem asks for a spanning tree with maximum total edge weight.
In the context of abstract graphs, the two problems are algorithmically equivalent in the sense that an algorithm for finding a Min-ST can also find a Max-ST within the same time bound (by simply negating the edge weights), and vice versa. The situation is quite different in the context of geometric graphs where vertices are points in the plane and edge wights are Euclidean distances between pairs of points. An algorithm that uses the geometry of the Euclidean plane for finding a Min-ST may not be useful for computing a Max-ST because there is no known geometric transformation for setting up a duality between the ``nearest'' and the ``farthest'' relations between points \cite{Monma1990}. In fact the existing geometric algorithms for the Min-ST and Max-ST problems exploit different sets of techniques.  

\begin{figure}[htb]
	\centering
	\setlength{\tabcolsep}{0in}
	$\begin{tabular}{cccc}
	\multicolumn{1}{m{.25\columnwidth}}{\centering\includegraphics[width=.17\columnwidth]{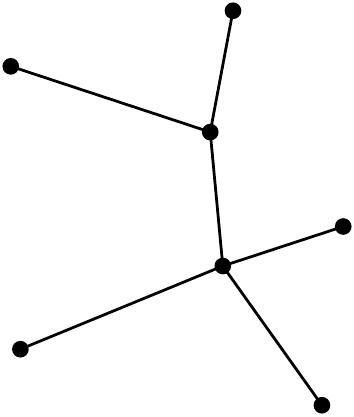}}
	&\multicolumn{1}{m{.25\columnwidth}}{\centering\vspace{0pt}\includegraphics[width=.17\columnwidth]{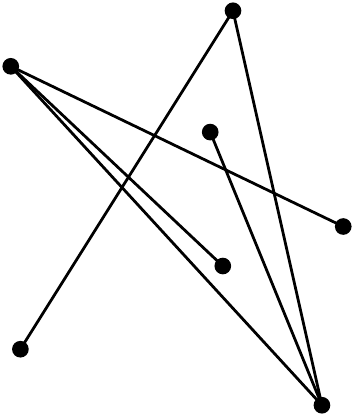}}
	&\multicolumn{1}{m{.25\columnwidth}}{\centering\includegraphics[width=.17\columnwidth]{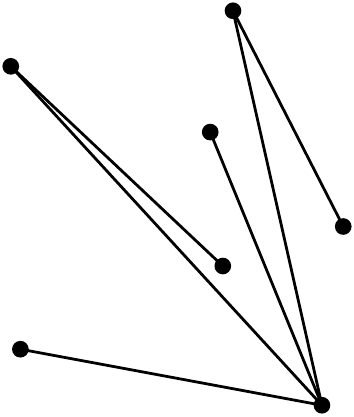}}
	&\multicolumn{1}{m{.25\columnwidth}}{\centering\vspace{0pt}\includegraphics[width=.19\columnwidth]{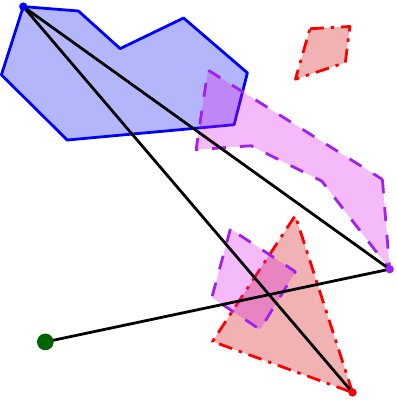}}
	\\
	(a) Min-ST&(b) Max-ST &(c) Max-NC-ST&(d) Max-ST-NB
	\end{tabular}$
	\caption{(a) minimum spanning tree, (b) maximum spanning tree, (c) longest noncrossing spanning tree, and (d) longest spanning tree with four neighborhoods colored red, green, blue, and purple.}
	\label{setting-fig}
\end{figure}

The problems of computing spanning trees with enforced properties (such as having minimum weight, maximum weight, bounded degree, or being noncrossing) have been well-studied in the last decades for both abstract graphs and geometric graphs.
We study two problems related to maximum spanning trees in geometric graphs. The maximum spanning tree and related problems, in addition to their fundamental nature, find applications in worst-case analysis of heuristics for various problems in combinatorial optimization \cite{Alon1995}, and in approminating maximum triangulations \cite[pp.~338]{Bern1996}. They also have applications in cluster analysis, where one needs to partition a set of entities into well-separated and homogeneous clusters \cite{Asano1988, Monma1990}. Maximum spanning trees are directly related to computing diameter and farthest neighbors which are fundamental problems in computational geometry, with many applications \cite{Agarwal1991}.

In the classical Euclidean Max-ST problem we are given a set of $n$ points in the plane (as vertices) and we want to find a spanning tree of maximum total edge length, where the length of every edge is the Euclidean distance between its two endpoints; see Figure~\ref{setting-fig}. This problem can be solved in $O(n^2)$ time by Prim's algorithm \cite{Fredman1987} for abstract graphs, and in $O(n\log n)$ time by an algorithm that uses the geometry \cite{Monma1990}. In contrast to the Euclidean Min-ST which is always noncrossing (because of the triangle inequality), the Euclidean Max-ST is almost always self-crossing. 
One problem that we study in this paper is the {\em longest noncrossing spanning tree} (Max-NC-ST) problem which is to compute a noncrossing spanning tree of maximum length, as depicted in Figure~\ref{setting-fig}. It is not known whether or not this problem is NP-hard.

Another problem that we study is the {\em longest spanning tree with neighborhoods} (Max-ST-NB): Given a collection of $n$ regions (neighborhoods) in the plane, we want to find a maximum-length tree that connects $n$ representative points, one point from each polygon, as in Figure~\ref{setting-fig}. We emphasis that the tree should contain exactly one point from each neighborhood.
Each {\em neighborhood} is the union of simple polygons, and the neighborhoods are not necessarily disjoint. The neighborhoods are assumed to be colored by $n$ different colors. 
The hardness of the Max-ST-NB problem is open. The difficulty lies in choosing the representative points; once these points are selected, the problem is reduced to the Euclidean Max-ST problem.
\subsection{Related work on the longest noncrossing spanning tree} 
Inspired by the seminal work of Alon, Rajagopalan, and Suri in SoCG 1993 \cite{Alon1995}, the study of long noncrossing configurations in the plane has received considerable attention in recent years. Alon~\etal~show how to compute constant-factor approximations of longest noncrossing spanning trees, perfect matchings, and Hamiltonian paths.
They show that the longest {\em star}, a tree in which one vertex is connected to all others, always gives a $0.5$ approximation of the longest spanning tree (a short proof of this claim is given in \cite{Dumitrescu2010}). As pointed out by Alon \etal~the ratio $0.5$ between the lengths of a longest star and a longest (possibly crossing) spanning tree is the best possible (in the limit); this can be verified by placing $n/2$ points at arbitrary small neighborhood around $(0,0)$ and $n/2$ points at arbitrary small neighborhood around $(1,0)$. Therefore, to obtain a better approximation ratio one should take into account spanning trees other than stars. The ratio $0.5$ remained the best known for almost seventeen years until Dumitrescu and
T{\'{o}}th (STACS 2010) \cite{Dumitrescu2010} slightly improved it to $0.502$, which was then improved to $0.503$ by Biniaz \etal~\cite{Biniaz2019}. The ratios $0.5$, $0.502$, and $0.503$ are obtained by considering the length of a longest (possibly crossing) spanning tree as the upper bound. Although such a tree provides a safe upper bound, it is not a valid solution for the Max-NC-ST problem. Alon~\etal~show that if the longest crossing spanning tree were to used as the upper bound then the approximation ratio cannot be improved beyond $2/\pi<0.637$ because such a tree can be $\pi/2$ times longer than a longest noncrossing spanning tree.
Recently, Cabello~\etal~\cite{Cabello2020} employed a longest noncrossing spanning tree as the upper bound and obtained a relatively significant improved ratio $0.512$.

The survey article by Eppstein \cite[pp.~439]{Eppstein2000} lists the hardness of Max-NC-ST as an open problem in the context of geometric network optimization. This problem has also been studied also for other structures. Alon \etal~show approximation ratios $2/\pi$ and $1/\pi$ for the longest noncrossing perfect matching and Hamiltonian path, respectively. The ratio for the Hamiltonian path was improved to $2/(1+\pi)$ by Dumitrescu and
T{\'{o}}th who also gave the first (tho non constant-factor) approximation algorithm for the longest noncrossing Hamiltonian cycle. The longest noncrossing spanning tree is also studied in multipartite geometric graphs \cite{Biniaz2019}. 

\subsection{Related work on the longest spanning tree with neighborhoods} 

The Max-ST-NB problem has the same flavor as the Euclidean group Steiner tree problem in which we are given $n$ groups of points in the plane and the goal is to find a shortest tree that contains ``at least'' one point from each group. The general group Steiner tree problem is NP-hard and cannot be approximated by a factor $O(\log^{2-\epsilon} n)$ for any $\epsilon>0$ \cite{Halperin2003}.
The Max-ST-NB problem also lies in the concept of imprecision in computational geometry where each input point provided as a region of uncertainty and the exact position of the point may be anywhere in the region; see e.g. \cite{Dorrigiv2015, Loffler2010}.

Similar to that of Max-NC-ST, one can show that the longest star, in which one vertex of a polygon is connected to one vertex in all other polygons, achieves a $0.5$-approximate solution for the Max-ST-NB problem. Recently, Chen and Dumitrescu \cite{Chen2018} present an approximation algorithm with improved ratio $0.511$. Although their algorithm is simple, the analysis of its ratio is thoroughly involved. They also show that the approximation ratio of an algorithm that always includes a bichromatic diametral pair as an edge in the solution cannot be better than $\sqrt{2-\sqrt{3}}\approx 0.517$. 

Analogous problem has been studied for other structures with neighborhoods, e.g., minimum spanning tree \cite{Blanco2017, Dorrigiv2015, Yang2007}, traveling salesman tour \cite{Arkin1994, Mitchell2007, Mitchell2010}, convex hull \cite{Loffler2010, Kreveld2008}, to name a few. 

\subsection{Our contributions and approach}

We report improved approximation ratios for the Max-NC-ST and the Max-ST-NB problems. Our results are summarized in the following theorems.

\begin{theorem}
	\label{neighborhood-thr}
	A $0.524$-approximation for the longest spanning tree with neighborhoods can be computed in linear time after computing a bichromatic diameter.
\end{theorem}

\begin{theorem}
	\label{noncrossing-thr}
	A $0.519$-approximation for the longest noncrossing spanning tree can be computed in polynomial time.
\end{theorem}
The new approximation ratios are obtained mainly by employing the Euclidean Steiner ratio, which has not been used in this context earlier. The employment is not straightforward. We use the Steiner ratio to handle a situation where some points lie far from the diameter of the input. This situation is a bottleneck of previous algorithms, for both problems. To handle this situation we first obtain a lower bound for the length of the minimum spanning tree of a small subset of the input. We use this lower bound, along with the Steiner ratio, and obtain a lower bound on the length of the Steiner minimal tree of the subset. Then we use this new lower bound to construct a long spanning tree on the entire input. The employment of the Steiner ratio, not only improves the approximation ratios but also simplifies the analysis. To see why this employment is nontrivial, one may think of this counterintuitive question: how could a lower bound on the length of the ``minimum'' spanning tree lead to a constant-ratio approximation of the ``maximum'' spanning tree? 

For the Max-NC-ST problem we give a polynomial-time approximation algorithm with improved ratio $0.519$. Following the successively improved trend of the approximation ratio from $0.5$ \cite{Alon1995}, to $0.502$ \cite{Dumitrescu2010}, $0.503$ \cite{Biniaz2019}, and to $0.512$ \cite{Cabello2020} shows that even a small improvement requires a significant effort. To obtain the new ratio we borrow some ideas from previous works and combine them with some new ideas (including the use of the Steiner ratio) along with a more refined analysis. The ratios obtained by our algorithm and that of \cite{Cabello2020} are with respect to the longest ``noncrossing'' spanning tree while the ratios of \cite{Alon1995, Biniaz2019, Dumitrescu2010} are with respect to the longest spanning tree.


For the Max-ST-NB problem we give an approximation algorithm with improved ratio $0.524$.
The algorithm is not complicated: we find a bichromatic diameter (farthest pair of input vertices with different colors) and use it to compute three stars and one {\em double-star} (a tree of diameter 3), and then report the longest one. After computing a bichromatic diameter (which is a well-studied problem) the rest of the algorithm takes linear time. Our analysis of the the ratio $0.524$ is relatively short, compared to that of Chen and Dumitrescu \cite{Chen2018} for the ratio $0.511$. The shortness comes form again the use of the Steiner ratio which takes care of a bottleneck situation (described above)---Chen and Dumitrescu devoted a thorough analysis for handling this situation. 

As a secondary result, for the Max-ST-NB problem we show that the approximation ratio of an algorithm that always includes a bichromatic diameter in the solution cannot be better than $0.5$, thereby improving the previous upper bound $0.517$ by Chen and Dumitrescu.

\section{Preliminaries for the algorithms}
\label{preliminaries}
Both our algorithms make extensive use of trees with low diameter, such as stars and double-stars. In fact our solutions for the Max-ST-NB and the Max-NC-ST problems have diameters at most three and six, respectively. A {\em star}, centered at a vertex $p$, is a tree in which every edge is incident to $p$. A {\em double-star}, centered at two vertices $p$ and $q$, is a tree that contains the edge $pq$ and its every other edge is incident to either $p$ or $q$.

We denote the Euclidean distance between two points $p$ and $q$ in the plane by $|pq|$. A {\em geometric graph} is a graph whose vertices are points in the plane and whose edges are straight line segments. The length of a geometric graph $G$, denoted by $\len{G}$, is the total length of its edges. We denote by $D(p,r)$ the disk of radius $r$ that is centered at point $p$.

For a point set $P$, a {\em diametral pair} is a pair of points that attain the maximum distance. If the points in $P$ are colored, then a {\em bichromatic diametral pair} is defined as a pair of points with different colors that attain the maximum distance.

The {\em Euclidean Steiner tree} problem asks for the shortest connected geometric graph spanning a given set $P$ of points in the plane. The solution takes the form of a tree, called a Steiner minimal tree (SMT), that includes all points of $P$ along with zero or more extra vertices called {\em Steiner points} \cite{Bern1996}. The {\em Steiner ratio} is defined to be the infimum of the length of the Steiner minimal tree divided by the length of the minimum spanning tree, over all point sets in the plane:
\begin{linenomath*}
\[\rho=\inf_{P\subset \mathbb{R}^2}\left\{\frac{\len{\textrm{SMT}(P)}}{\len{\textrm{Min-ST}(P)}}\right\}.\]
\end{linenomath*}
An old conjecture of Gilbert and Pollak \cite{Gilbert1968} states that $\rho=\sqrt{3}/{2}\approx 0.866$; this is achieved when $P$ is the vertices of an equilateral triangle. Although the conjecture seems to be still open \cite{Innami2010, Ivanov2012}, it has been verified verified when $P$ has $3$ \cite{Gilbert1968}, $4$ \cite{Pollak1978}, $5$ \cite{Du1985}, or $6$ \cite{Rubinstein1991} points. For the purpose of our algorithms the original proof of Gilbert and Pollak for $|P|=3$ is enough.

\paragraph{A simple $0.5$-approximation algorithm.} Chen and Dumitrescu~\cite{Chen2018} gave the following simple $0.5$-approximation algorithm for the Max-ST-NB (a similar approach was previously used in \cite{Dumitrescu2010}).
Take a bichromatic diametral pair $(a,b)$ from the given $n$ neighborhoods; $a$ and $b$ belong to two different neighborhoods. Each edge of every optimal solution $T^*$ has length at most $|ab|$, and thus $\len{T^*}\leqslant (n-1)|ab|$. Pick an arbitrary point $p$ from each of the the other $n-2$ neighborhoods. Let $S_a$ be the star obtained by connecting $a$ to $b$ and all points $p$. Define $S_b$ analogously. By the triangle inequality  $\len{S_a}+\len{S_b}\geqslant n |ab|\geqslant \len{T^*}$. Therefore the longer of $S_a$ and $S_b$ is a $0.5$-approximate solution for the Max-ST-NB problem. This idea also achieves a $0.5$-approximate solution for the Max-NC-ST problem (for which every input point can be viewed as a neighborhood).

\section{Maximum spanning tree with neighborhoods}
\label{neighborhood-section}
In this section we prove Theorem~\ref{neighborhood-thr}. Put $\delta = 0.524$.

We describe our algorithm for the Max-ST-NB first, as it easier to understand. It also gives insights on a better understanding of our algorithm for the Max-NC-ST problem.
To facilitate comparisons we use the same notation as of Chen and Dumitrescu \cite{Chen2018}. Let $\mathcal{X}=\{X_1,X_2,\dots, X_n\}$ be the given collection of $n$ polygonal neighborhoods of $N$ total vertices. We may assume that each $X_i$ is colored by a unique color. We present a $\delta$-approximation algorithm for computing a longest spanning tree with neighborhoods in $\mathcal{X}$.
Our algorithm selects representative points only from boundary vertices of polygonal neighborhoods. Thus, in the algorithm (but not in the analysis) we consider each polygonal neighborhood $X_i$ as the set of its boundary vertices, and consequently we consider $\mathcal{X}$ as a collection of $N$ points colored by $n$ colors. 
Define the {\em longest spanning star} centered at a vertex $p\in X_i$ as the star connecting $p$ to its farthest vertex in every other neighborhood.

\paragraph{Algorithm.}The main idea of the algorithm is simple: we compute a spanning double-star $D$ and three spanning stars $S_1, S_2, S_3$, and then report the longest one. 

Let $(a,b)$ be a bichromatic diametral pair of $\mathcal{X}$. 
After a suitable relabeling assume that $a \in X_1$ and $b\in X_2$. We compute $D$ as follows. 
Add the edge $ab$ to $D$. For each $X_i$, with $i\in\{3,\dots,n\}$, find a vertex $p_i\in X_i$ that is farthest from $a$ and find a vertex $q_i\in X_i$ that is farthest from $b$ (it might be the case that $p_i=q_i$). If $|ap_i|\geqslant |bq_i|$ then add $ap_i$ to $D$ otherwise add $bq_i$ to $D$. Observe that $D$ spans all neighborhoods in $\mathcal{X}$, and each edge of $D$ has length at least $|ab|/2$. Now we introduce the three stars.
Let $a'$ be a vertex of $X_1$ that is farthest from $a$, and let $b'$ be a vertex of $X_2$ that is farthest from $b$. Notice that $a'$ has the same color as $a$, and $b'$ has the same color as $b$. We compute $S_1$ as the longest spanning star that is centered at $a'$, and we compute $S_2$ as the longest spanning star that is centered at $b'$. 
Now let $c$ be a vertex in $\mathcal{X}$ that maximizes $|ac|+|bc|$. The vertex $c$ can be in any of neighborhoods $X_1,\dots,X_n$, also it might be the case that $c=a'$ or $c=b'$. We compute $S_3$ as the longest spanning star that is centered at $c$. 

\paragraph{Running time.}  It is implied by a result of Biniaz~\etal~\cite{Biniaz2018} that a bichromatic diametral pair of $\mathcal{X}$ can be found in $O(N\log N\log n)$ time (the algorithm of Bhattacharya and
Toussaint \cite{Bhattacharya1983} also computes a bichromatic diameter, but only for two-colored points). After finding $(a,b)$, the rest of the algorithm (finding $a', b', c$, and finding farthest points from $a$, $b$, $a'$, $b'$, $c$) takes $O(N)$ time.

\subsection{Analysis of the approximation ratio}
For the analysis we consider $\mathcal{X}$ as the initial collection of polygonal neighborhoods. Let $T^*$ denotes a longest spanning tree with neighborhoods in $\mathcal{X}$.
It is not hard to see that for any point in the plane, its farthest point in a polygon $P$ must be a  vertex of $P$ (see e.g. \cite[Chapter 7]{deBerg2008}). Thus, any bichromatic diameter of $\mathcal{X}$ is introduced by two vertices of polygons in $\mathcal{X}$. Hence the pair $(a,b)$, selected in the algorithm, is a bichromatic diametral pair of the initial collection $\mathcal{X}$. Therefore, $|ab|$ is an upper bound for the length of edges in $T^*$. Recall our assumption from the algorithm that $a\in X_1$ and $b\in X_2$.
After a suitable rotation, translation, and scaling assume that $ab$ is horizontal, $a=(0,0)$, and $b=(1,0)$. Since $|ab|=1$ and $T^*$ has $n-1$ edges,
\begin{linenomath*} 
\begin{equation}
\label{eq-upperbound-1}
\len{T^*}\leqslant (n-1)|ab|\leqslant n-1.
\end{equation}
\end{linenomath*}
\begin{lemma}
	\label{Sa-Sb}
	The double star $D$ is longer than any star that is centered at $a$ or at $b$.
\end{lemma}
\begin{proof}
Due to symmetry, we only prove this lemma for any star $S_a$ centered at $a$.
Recall that $D$ contains the edge $ab$ where $b\in X_2$, and for each $i\in\{3,\dots, n\}$ it contains the longer of $ap_i$ and $aq_i$ where $p_i$ and $q_i$ are the farthest vertices of $X_i$ from $a$ and $b$, respectively. 

For any $i\in\{2,\dots, n\}$ let $r_i$ be the vertex of $X_i$ that is connected to $a$ in $S_a$. If $i=2$ then $|ar_i|\leqslant |ab|$. If $i>2$ then $|ar_i|\leqslant \max\{|ap_i|,|bq_i|\}$. Therefore $\len{D}\geqslant \len{S_a}$.
\end{proof}

Recall $a'$ and $b'$ as vertices of $X_1$ and $X_2$ that are farthest from $a$ and $b$, respectively. 

\begin{lemma}
	\label{aa-bb-large}
	If $|aa'|\geqslant 2\delta$ or $|bb'|\geqslant 2\delta$ then $\max\{\len{S_1},\len{S_2},\len{D}\}\geqslant \delta\cdot \len{T^*}$.
\end{lemma}

\begin{proof}
First assume that $|aa'|\geqslant 2\delta$.
Choose a vertex $p_i$ in each neighborhood $X_i$, with $i\neq1$, and connect it to $a$ and to $a'$. We obtain two stars $S_a$ and $S_{a'}$ that are centered at $a$ and $a'$. The total length of these stars is at least
\begin{linenomath*} \[\sum_{i\in\{2,\dots,n\}}(|ap_i|+|p_ia'|)\geqslant \sum_{i\in\{2,\dots,n\}}|aa'|\geqslant 2\delta (n-1),\]
\end{linenomath*} and thus the longer star has length at least $\delta(n - 1)$. If $S_{a'}$ is longer then $\len{S_1}\geqslant \len{S_{a'}}\geqslant \delta(n - 1)\geqslant \delta\cdot\len{T^*}$, where the first inequality is implied by the fact that $S_1$ is the longest spanning star centered at $a'$ and the last inequality is implied by \eqref{eq-upperbound-1}. If $S_{a}$ is longer then  $\len{D}\geqslant \len{S_{a}}\geqslant \delta(n - 1)\geqslant \delta\cdot\len{T^*}$, where the first inequality is implied by Lemma~\ref{Sa-Sb}.

If $|bb'|\geqslant 2\delta$, an analogous implication show that the length of $S_2$ or $D$ is at least $\delta\cdot\len{T^*}$.
\end{proof}

Having Lemma~\ref{aa-bb-large} in hand, in the rest if this section we turn our attention to the case where $|aa'|<2\delta$ and $|bb'|<2\delta$. 
The intersection of two disks is called a {\em lens}. Define lens $L= D(a,1)\cap D(b,1)$ to be the region of distance at most $1$ from $a$ and $b$. Since $(a,b)$ is a bichromatic diametral pair and $|ab|=1$, all vertices of $X_3,\dots, X_n$ lie in $L$, all vertices of $X_1$ lie in $D(b,1)$, and all vertices of $X_2$ lie in $D(a,1)$. Moreover since $|aa'|<2\delta$ all vertices of $X_1$ lies in lens $L_1=D(b,1)\cap D(a,2\delta)$, and since $|bb'|<2\delta$  all vertices of $X_2$ lie in lens $L_2=D(a,1)\cap D(b,2\delta)$. See Figure~\ref{lens}(a) for an illustration. In this setting, all vertices of $\mathcal{X}$ lie in $L_1\cup L_2$.

\begin{figure}[htb]
	\centering
	\setlength{\tabcolsep}{0in}
	$\begin{tabular}{cc}
	\multicolumn{1}{m{.5\columnwidth}}{\centering\includegraphics[width=.34\columnwidth]{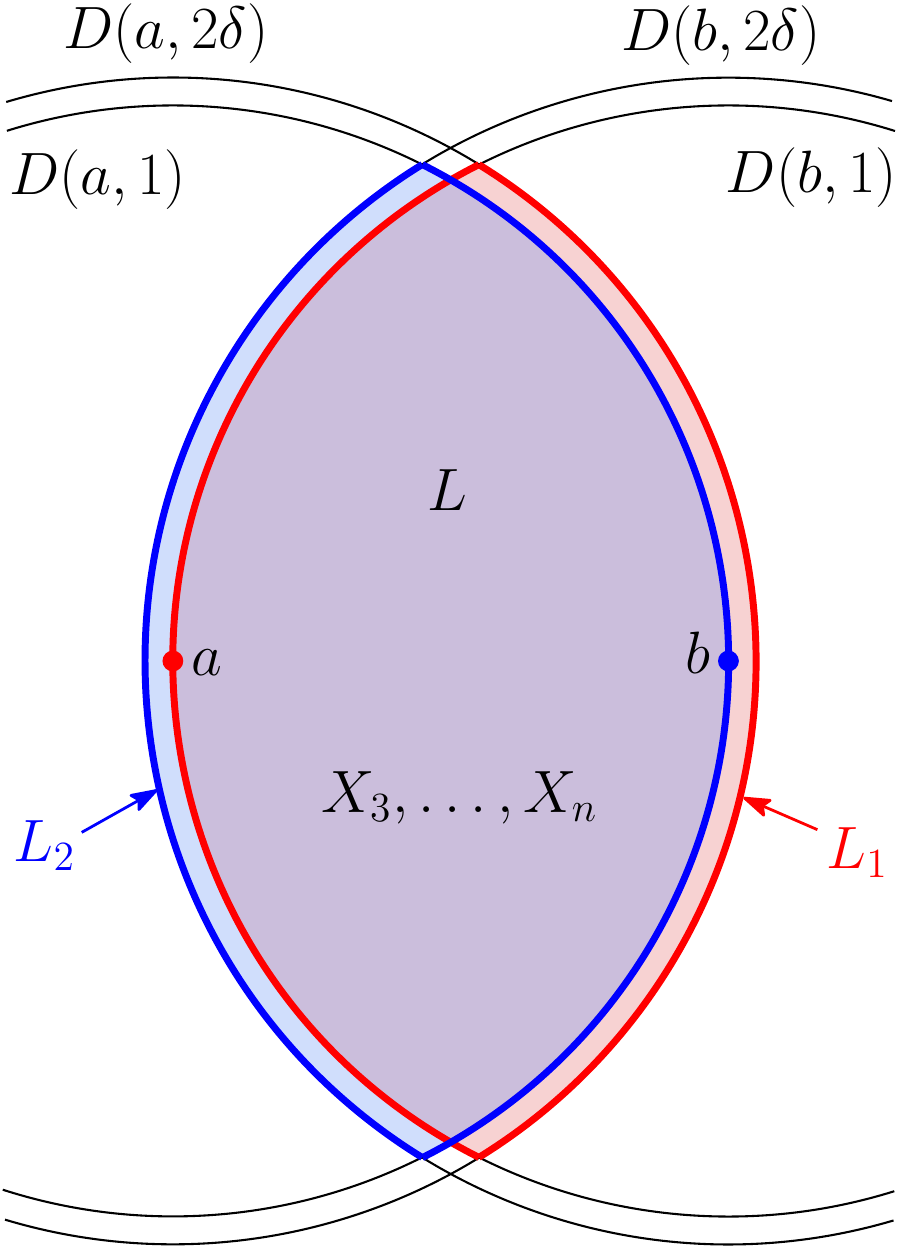}}
	&\multicolumn{1}{m{.5\columnwidth}}{\centering\vspace{0pt}\includegraphics[width=.4\columnwidth]{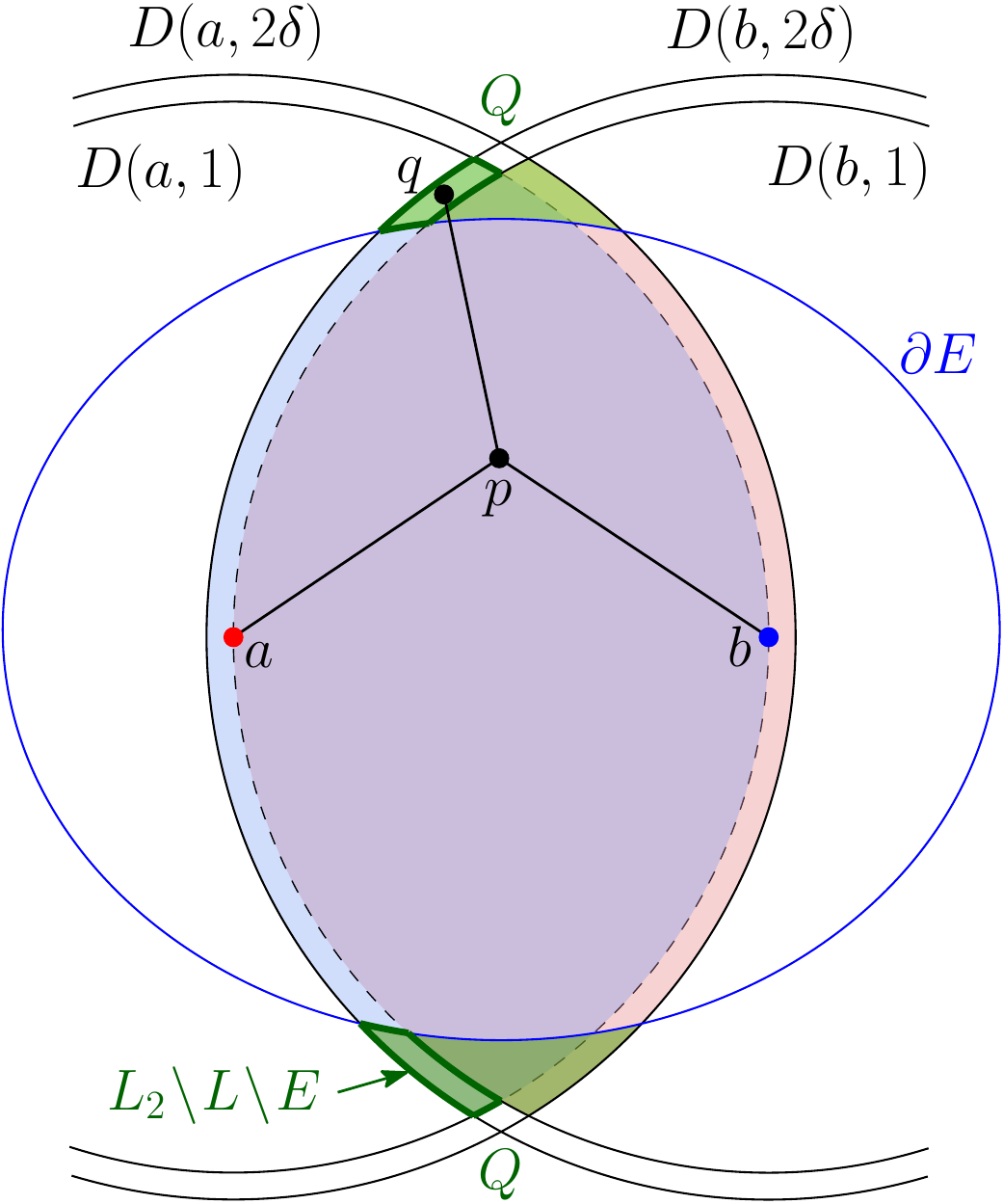}}
	\\
	(a)&(b)
	\end{tabular}$
	\caption{Illustration of (a) lenses $L$, $L_1$, $L_2$, and (b) the ellipse $\partial E$ and the region $Q$.}
	\label{lens}
\end{figure}

We fix a parameter $\omega=\frac{6\delta}{\sqrt{3}}-1\approx 0.815$.
Let $E$ be the set of all points whose total distance from $a$ and $b$ is at most $\omega+2\delta$, i.e., $E=\{x\in\mathbb{R}^2: |xa|+|xb|\leqslant \omega+2\delta\}$. The boundary $\partial E$ of $E$ is an ellipse with foci $a$ and $b$. Put $Q=(L_1
\cup L_2)\!\setminus\! E$, as depicted in Figure~\ref{lens}(b).  

\begin{lemma}
	\label{qa-qb}
	For any point $q\in Q$ it holds that $|aq|>\omega$ and $|bq|>\omega$.
\end{lemma}
\begin{proof}
The point $q$ lies in $L_1\cup L_2$, and thus $|aq|\leqslant 2\delta$ and $|bq|\leqslant 2\delta$. The point $q$ does not lie in $E$, and hence $|aq|+|bq|>\omega+2\delta$. Combining these inequalities yields $|aq|>\omega$ and $|bq|>\omega$.
\end{proof}

The next ``helper lemma'' is the place where we use the Steiner ratio to obtain a lower bound on the length of the Steiner minimal tree of a subset of input vertices.

\begin{lemma}
	\label{Steiner-lemma}
	For any two points $q\in Q$ and $p\in\mathbb{R}^2$ it holds that $|pa|+|pb|+|pq|> 3\delta$. 
\end{lemma}

\begin{proof}
Put $P=\{a,b,q\}$. We are going to obtain a lower bound for the Steiner minimal tree of $P$ by the minimum spanning tree of $P$. The point $q$ lies in $L_2\!\setminus\! L\!\setminus\! E$ or in $L_1\!\setminus\! L\!\setminus\! E$ or in 
$L\!\setminus\! E$.
If $q\in L_2\!\setminus\! L\!\setminus\! E$, then $bq$ is the longest side of the triangle $\bigtriangleup(a,b,q)$; this case is depicted in Figure~\ref{lens}(b). In this case Min-ST($P$) has edges $aq$ and $ab$. Thus, $\len{\textrm{Min-ST}(P)}=|aq|+|ab|> \omega + 1$ where the inequality holds by Lemma~\ref{qa-qb} and the fact that $|ab|=1$. If $q\in L_1\!\setminus\! L\!\setminus\! E$, then analogously we get $\len{\textrm{Min-ST}(P)}> \omega+1$. If
$q\in L\!\setminus\! E$, then $ab$ is the longest side of the triangle $\bigtriangleup(a,b,q)$, and hence  Min-ST($P$) has edges $aq$ and $bq$. Thus, $\len{\textrm{Min-ST}(P)}=|aq|+|bq|> \omega+2\delta>\omega+1$ where the first inequality holds because $q\notin E$. Therefore, in all cases we have $\len{\textrm{Min-ST}(P)}> \omega+1$.

The union of the three segments $pa$, $pb$, and $pq$ form a tree that connects the points of $P$, as in Figure~\ref{lens}(b). The length of this tree cannot be smaller than the length of Steiner minimal tree of $P$, and thus $|pa|+|pb|+|pq|\geqslant \len{\textrm{SMT}(P)}$.
By using the Steiner ratio for three points (proved by Gilbert and Pollak \cite{Gilbert1968}) we get 
\begin{linenomath*}
\[|pa|+|pb|+|pq|\geqslant \len{\textrm{SMT}(P)}\geqslant \frac{\sqrt{3}}{2}\cdot\len{\textrm{Min-ST}(P)}> \frac{\sqrt{3}}{2}\cdot (\omega+1)=3\delta.\qedhere
\]\end{linenomath*}
\end{proof}

The following two lemmas consider two cases depending on whether or not a vertex of $\mathcal{X}$ lies in $Q$. Both lemmas benefit from our helper Lemma~\ref{Steiner-lemma}. In Lemma~\ref{Q-not-empty} we use the helper lemma directly to obtain a lower bound on the maximum length of $S_3$ and $D$. In Lemma~\ref{Q-empty} we use the helper lemma indirectly to obtain a better upper bound on the length of $T^*$.

\begin{lemma}
	\label{Q-not-empty}
	If at least one vertex of $\mathcal{X}$ lies in $Q$ then $\max\{\len{S_3},\len{D}\}\geqslant \delta\cdot \len{T^*}$.
\end{lemma}
\begin{proof}
Let $q$ be any vertex of $\mathcal{X}$ in $Q$.
We have three cases: (i) $q\notin X_1$ and $q\notin X_2$, (ii) $q\in X_1$, (iii) $q\in X_2$. First consider case (i). After a suitable relabeling assume that $q\in X_3$. Consider an arbitrary representative vertex $p_i$ from each $X_i$ with $i\in\{4,\dots,n\}$. It is implied by  Lemma~\ref{Steiner-lemma} that
\begin{linenomath*}
\[\sum_{i=4}^{n}|p_ia|+|p_ib|+|p_iq|>3\delta(n-3).\]
\end{linenomath*}
Let $x$ denotes the point in $\{a,b,q\}$ that has the largest total distance to all points $p_i$. By bounding the maximum with the average, the total distance of $p_i$'s to $x$ is at least $\delta(n-3)$. If $x=q$, then the star that connects $q$ to $p_4, \dots, p_n$, $a$, and $b$ has length  \begin{linenomath*}
	\[|qp_4|+\dots+|qp_n|+|qa|+|qb|>\delta(n-3)+ \omega + \omega> \delta (n-1)\geqslant \delta\cdot\len{T^*},\]
\end{linenomath*} where the inequalities hold by Lemma~\ref{qa-qb}, that fact that $w>\delta$, and \eqref{eq-upperbound-1}. In this case $\len{S_3}\geqslant \delta\cdot\len{T^*}$, where the vertex $c$ in the algorithm plays the role of $q$ in the analysis. If $x=a$, then the star that connects $a$ to $p_4, \dots, p_n$, $q$, and $b$ has length
\begin{linenomath*} \[|ap_4|+\dots+|ap_n|+|aq|+|ab|>\delta(n-3)+ \omega + 1>\delta(n-1)\geqslant \delta\cdot\len{T^*}.\] 
\end{linenomath*}The length of this star is not larger than the length of $D$ (by Lemma~\ref{Sa-Sb}), and thus $\len{D}\geqslant \delta\cdot\len{T^*}$. If $x=b$, an analogous argument implies that $D$ is a desired tree, proving the lemma for case (i).
	
Now consider case (ii) where $q\in X_1$. Our proof of this case is somewhat similar to that of case (i). Consider an arbitrary representative vertex $p_i$ from each $X_i$ with $i\in\{3,\dots,n\}$.  It is implied by  Lemma~\ref{Steiner-lemma} that
$\sum_{i=3}^{n}|p_ia|+|p_ib|+|p_iq|>3\delta(n-2).$
Let $x$ denotes the point in $\{a,b, q\}$ that has the largest total distance to all points $p_i$. The total distance of $p_i$'s to $x$ is at least $\delta(n-2)$. If $x=q$, then the star connecting $q$ to $p_3, \dots, p_n$, and $b$ has length at least $\delta(n-2)+ \omega> \delta (n-1)\geqslant \delta\cdot\len{T^*}$. In this case $\len{S_3}\geqslant \delta\cdot\len{T^*}$, where $c$ plays the role of $q$. If $x=a$, then the star connecting $a$ to $p_3, \dots, p_n$, and $b$ has length at least $\delta(n-2)+1> \delta\cdot\len{T^*}$. The length of this star is not larger than that of $D$, and thus $\len{D}\geqslant \delta\cdot\len{T^*}$. If $x=b$, by an analogous argument $D$ is a desired tree, proving the lemma for case (ii). 
Our proof of case (iii) is analogous to that of case (ii).
\end{proof}
\begin{lemma}
	\label{Q-empty}
	If no vertex of $\mathcal{X}$ lies in $Q$ then $\len{D}\geqslant \delta\cdot \len{T^*}$.
\end{lemma}
\begin{proof}
In this case all vertices of $\mathcal{X}$ lie in the region $R=(L_{1} \cup L_{2}) \!\setminus\! Q$, as illustrated in Figure~\ref{double-lens}. 
Define the lens $L'=D(a,\delta)\cap D(b,\delta)$. Denote by $l$ the lowest point of $L'$. Disregarding symmetry, denote by $f$ the topmost intersection point of boundaries of $D(b,2\delta)$ and $E$ as in Figure~\ref{double-lens}. Consider the smallest disk $D_l$ with center $l$ that contains the entire region $R$. By basic geometry of circle-circle intersection (that boundaries of $D_l$ and $D(b,2\delta)$ intersect at two points) and circle-ellipse intersection (that boundaries of $D_l$ and $E$ intersect at four points---these four points are marked in Figure~\ref{double-lens}) we can verify that $D_l$ passes through $f$. Thus $f$ is the farthest point in $R$ from $l$. Conversely, it follows from circle-circle intersection that $l$ is the farthest point of $L'$ from $f$. Therefore $|lf|$ is an upper bound for the distance between any point in $L'$ and any point in $R$, and so is an upper bound on the length of any edge of $T^*$ with an endpoint in $L'$.

\begin{figure}[htb]
	\centering
	\includegraphics[width=.5\columnwidth]{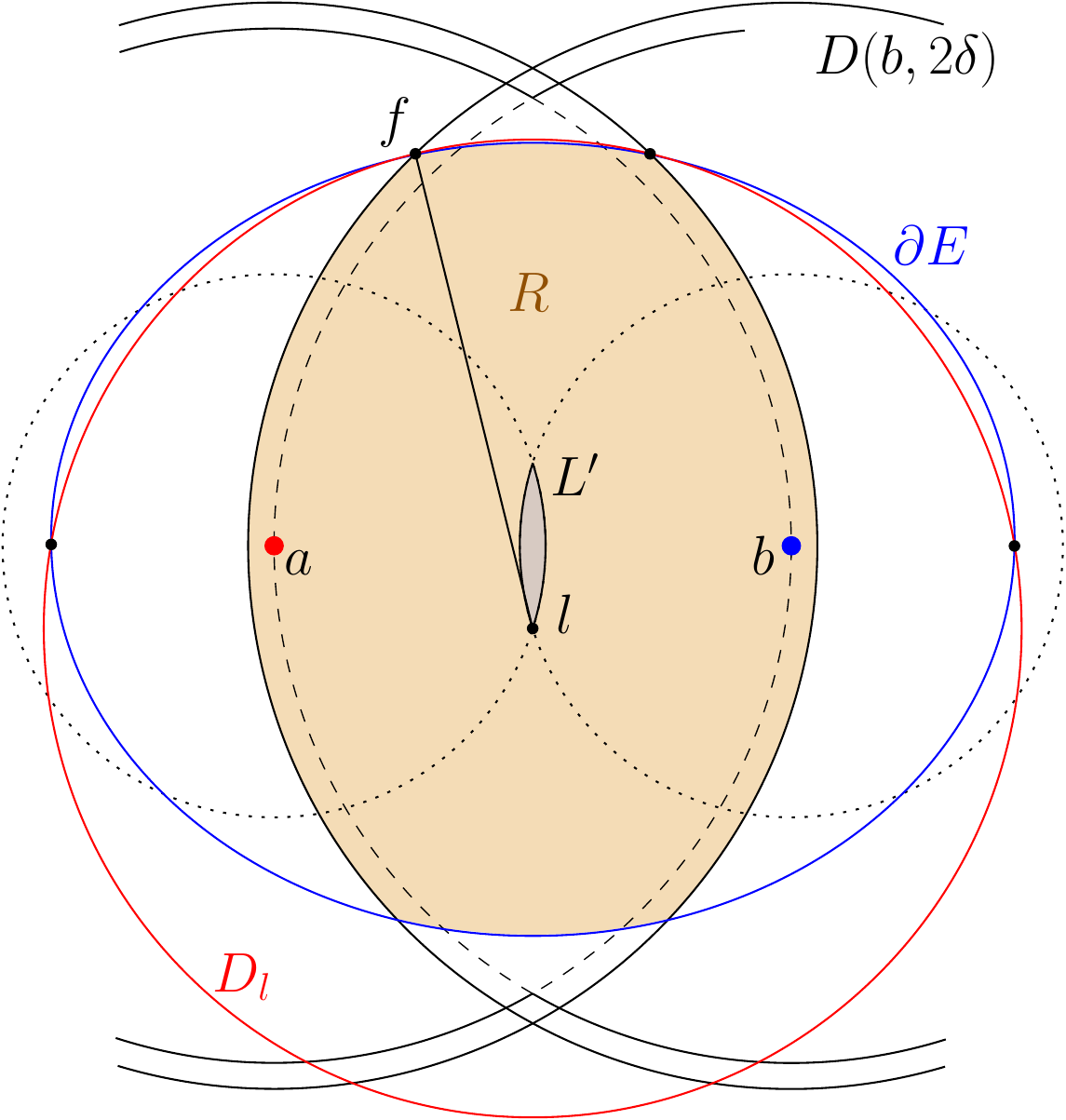}
	\caption{The length $|lf|$ is the largest possible for any edge having an endpoint in $L'$.}
	\label{double-lens}
\end{figure}

Since $f$ lies on $\partial E$ we have $|af|+|bf|=\omega+2\delta$, and since it lies on the boundary of $D(b,2\delta)$ we have $|bf|=2\delta$. Thus, $|af|=w$.
Therefore $f$ is an intersection point of boundaries of $D(a,\omega)$ and $D(b,2\delta)$ that are centered at $(0,0)$ and $(1,0)$, respectively. The point $l$ is an intersection point of boundaries of $D(a,\delta)$ and $D(b,\delta)$. By using the Pythagorean theorem and the circle equation, we can obtain the coordinates of $f$ and $l$, and give the following expression for $|lf|$:
\begin{linenomath*}
\[|lf|=\sqrt{\left(\frac{\omega^2-4\delta^2}{2}\right)^2 + \left(\sqrt{\omega^2-\left(\frac{1+\omega^2-4\delta^2}{2}\right)^2}+\sqrt{\delta^2-\frac{1}{4}}\right)^2}\approx0.9464<0.95.\]
\end{linenomath*}
We use $|lf|$ to obtain a better upper bound on the length of $T^*$.
Let $m$ be the number of neighborhoods $X_i$ that lie entirely in the interior of $L'$. Notice that $0\leqslant m\leqslant n-2$ ($X_1$ and $X_2$ do not lie in $L'$). The number of edges of $T^*$, that are incident to these neighborhoods, is at least $m$. Since each such edge is of length at most $|lf|$ ($<0.95$), we get
\begin{linenomath*}
\begin{equation}
\label{better-up}
\len{T^*}\leqslant (n-m-1)
+0.95 m=n -0.05m-1.
\end{equation}
\end{linenomath*}
Now we compute a lower bound on the length of $D$ in terms of $m$. One vertex of each neighborhood in $L'$ is connected to $a$ or $b$ (whichever is the farthest) in $D$. Each such connection has length at least $0.5$. There are $n-m-2$  neighborhoods that have vertices outside $L'$ (excluding $X_1$ and $X_2$). One vertex of each such neighborhood is connected to $a$ or $b$ (whichever is the farthest) in $D$. Each such connection has length at least $\delta$. The points $a$ and $b$ are connected to each other in $D$. Thus
\begin{linenomath*}
\begin{align}
\notag\len{D}&\geqslant \delta(n-m-2)+0.5m+1\\ \notag &=0.524n-0.024m-0.048\\\notag & >0.524(n
-0.05m-1)\\\notag &\geqslant 0.524\cdot \len{T^*},
\end{align}
\end{linenomath*}
where the equality holds by plugging $\delta=0.524$, and the last inequality holds by \eqref{better-up}. 
\end{proof}

The cases considered in Lemmas~\ref{aa-bb-large}, \ref{Q-not-empty}, and \ref{Q-empty} ensure that the length of one of $S_1$, $S_2$, $S_3$, and $D$ is at least $\delta\cdot\len{T^*}$. 
This concludes our analysis and proof of Theorem~\ref{neighborhood-thr}. 
\paragraph{Remark.} Although the stars $S_1$, $S_2$, $S_3$ are noncrossing, the double-star $D$ can have crossing edges. Thus the algorithm of this section cannot be used for the Max-NC-ST problem.
\paragraph{Inclusion of bichromatic diameter.} The simple $0.5$-approximation algorithm of Chen and Dumitrescu \cite{Chen2018} (described in Section~\ref{preliminaries}) always includes the bichromatic diametral pair $(a,b)$ in the solution. Chen and Dumitrescu argue that the approximation ratio of an algorithm that always includes a bichromatic diametral pair in the solution cannot be larger than $\sqrt{2-\sqrt{3}}\approx 0.517$.  We present an input instance that improves this upper bound to $0.5$, thereby showing that the ratio of the simple algorithm is tight, in this sense.

\begin{figure}[htb]
	\centering
	\includegraphics[width=.48\columnwidth]{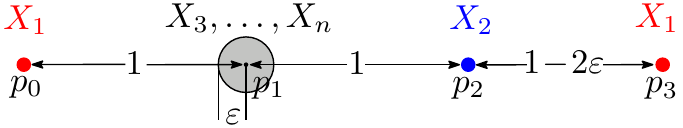}
	\caption{Illustration of the upper bound $0.5$ for inclusion of a bichromatic diametral pair.}
	\label{diametral-pair}
\end{figure} 

Consider four points $p_0=(0,0)$, $p_1=(1,0)$, $p_2=(2,0)$, and $p_3=(3-2\varepsilon,0)$ for arbitrary small $\varepsilon > 0$, e.g. $\varepsilon=1/n$. Our input instance consists of neighborhoods $X_1,\dots, X_n$ where $X_1=\{p_0, p_3\}$, $X_2=\{p_2\}$, and each of $X_3,\dots,X_n$ has exactly one point that is placed at distance at most $\varepsilon$ from $p_1$; see Figure~\ref{diametral-pair}. In this setting, $(p_0,p_2)$ is the unique bichromatic diametral pair. Consider any tree $T$ that contains the bichromatic diameter $p_0p_2$ (this means that $p_3$ is not in $T$). Any edge of $T$ incident to $X_3,\dots,X_4$ has length at most $1+\varepsilon$. Therefore $\len{T}\leqslant 2+(1+\varepsilon)(n-2)<n+1$. Now consider the tree $T^*$ that does not contain $p_0p_2$ but connects each of $X_2,\dots,X_n$ to $p_3$. The length of $T^*$ is at least $(1-2\varepsilon)+(2-3\varepsilon)(n-2)> 2n-6$. This would establish the upper bound $0.5$ on the approximation ratio because $\frac{\len{T}}{\len{T^*}}\leqslant \frac{n+1}{2n-6}$ which tends to $1/2$ in the limit.

\section{Maximum noncrossing spanning tree}
In this section we prove Theorem~\ref{noncrossing-thr}. Put $\delta = 0.519$.

Our $\delta$-approximation algorithm for the Max-NC-ST problem borrows some ideas from previous works \cite{Alon1995, Biniaz2018, Cabello2020, Dumitrescu2010} and combines them with some new ideas (including the use of the Steiner ratio) along with a more refined analysis. For such a well-studied problem, neither coming up with new ideas nor their combination with previous ideas is an easy task. 
We try to keep the proof short, but self-contained; we give a detailed description of our new ideas and a short description of borrowed ideas. To facilitate comparisons,  we give a proper citation for each part of the proof that overlaps with previous works, and we use the same notation as in the most recent related work \cite{Cabello2020}. 

Let $P\!\subset\! \mathbb{R}^2$ be the given point set of size $n$ and let $(u,v)$ be a diametral pair of $P$, throughout this section. After a suitable scaling assume that $|uv|=1$. Let $T^*$ be a longest noncrossing spanning tree of $P$. If $ab$ is a longest edge of $T^*$ then $|ab|\leqslant|uv|=1$, and thus 
\begin{linenomath*}
\begin{equation}
\label{T-star-bound}\len{T^*}\leqslant (n-1)|ab|\leqslant (n-1)|uv|\leqslant n-1. 
\end{equation}
\end{linenomath*}
	
Our plan is to construct a noncrossing spanning tree for $P$ of length at least $\delta\cdot\len{T^*}$. For a point $p\in P$, we denote by $S_p$ the star that connects $p$ to all other points of $P$. We start with the following simple lemma, proved in \cite{Dumitrescu2010}, which comes in handy in our construction.

\begin{lemma}
	\label{two-star-lemma}
	For any two points $p$ and $q$ in $P$ it holds that $\max\{\len{S_p},\len{S_q}\}\geqslant \frac{n}{2}|pq|$.
\end{lemma}
\begin{proof}
	By bounding the maximum with the average and then using the triangle inequality we get:
	\begin{linenomath*}
	\[\max\left\{\len{S_p},\len{S_q}\right\}\geqslant \frac{1}{2}\left(\len{S_p}+\len{S_q}\right)= \frac{1}{2}\sum_{r\in P}(|pr|+|rq|)\geqslant \frac{1}{2}\sum_{r\in P}|pq|= \frac{n}{2}|pq|.\quad\qedhere
	\]
	\end{linenomath*}
\end{proof}

From Lemma~\ref{two-star-lemma} and \eqref{T-star-bound} we have $\max\{\len{S_u},\len{S_v}\}\geqslant\frac{n}{2}|uv|=\frac{n}{2}>\frac{1}{2}\cdot\len{T^*}$. Therefore, the longer of $S_u$ and $S_v$ is a $0.5$ approximation of $T^*$. As pointed out by Alon \etal~\cite{Alon1995} the longest star may not give an approximation ratio better than $0.5$. To establish better ratios we need to consider trees that are not necessarily stars, and we need to incorporate more powerful ingredients.

Now we describe our $\delta$-approximation algorithm. It uses the noncrossing property of the optimal tree $T^*$ (similar to that of Cabello~\etal~\cite{Cabello2020}).

\paragraph{Algorithm approach.}Guess a longest edge $ab$ of $T^*$, say by trying all pairs of points in $P$. For each guess $ab$, construct seven noncrossing spanning trees as described in the rest of this section. Then report the longest tree, over all guesses for $ab$. 
\vspace{15pt}

From now on we assume that $ab$ is a longest edge of $T^*$. In the rest of this section we describe how to construct the seven noncrossing spanning trees in such a way that the length of the longest tree is at least $\delta\cdot\len{T^*}$. Fix a parameter $d=\frac{1}{2\delta}$. 

\begin{lemma}
	\label{ab-small}
	If $|ab|\leqslant d$ then $ \max\left\{\len{S_u},\len{S_v}\right\}\geqslant \delta\cdot\len{T^*}$.
\end{lemma}
\begin{proof}
By Lemma~\ref{two-star-lemma}, the fact that $|uv|=1=2\delta\cdot d$, our assumption that $|ab|\leqslant d$, and \eqref{T-star-bound} we get
\begin{linenomath*} \[\max\left\{\len{S_u},\len{S_v}\right\}\geqslant \frac{n}{2}|uv|=\delta \cdot nd\geqslant \delta \cdot n|ab|>\delta\cdot\len{T^*}.\qedhere\]
\end{linenomath*}\end{proof}

Having Lemma \ref{ab-small} in hand, in the rest of this section we assume that $d\leqslant |ab|\leqslant 1$. 
After a suitable rotation and a translation assume that $a=(0,0)$ and $b=(|ab|,0)$. Define the lens $L=D(a,1)\cap D(b,1)$; see Figure~\ref{noncrossing-fig1}(a). Since the diameter of $P$ is $1$, all points of $P$ lie in $L$. Fix parameters $\omega=0.16$ and $\hat{\beta}=0.44$ and define 
\begin{linenomath*}\[\hat{\alpha}=\frac{2\delta+3\omega-2}{\omega-1}-\hat{\beta}\quad\quad\quad\lambda=\frac{6\delta}{\sqrt{3}}+1-|ab|\quad\quad\quad\gamma=\frac{(2\delta+\hat{\alpha}-1)|ab|}{\hat{\alpha}}.\]\end{linenomath*}
We use the parameter $\lambda$ along with the Steiner ratio to take care of the situation where some input points lie far from $a$ and $b$---this situation is a bottleneck for previous algorithms.

Define $E_1=\{x\in\mathbb{R}^2: |xa|+|xb|\leqslant \lambda\}$. The boundary $\partial E_1$ is an ellipse with foci $a$ and $b$. Put $Q=L\!\setminus\! E_1$, as depicted in Figure~\ref{noncrossing-fig1}(a). Lemma~\ref{Steiner-lemma-2}, below, plays an important role in improving the approximation ratio. We keep its proof short as it is somewhat similar to our proof of Lemma~\ref{Steiner-lemma}. 

\begin{lemma}
	\label{qa-qb-2}
	For any point $q\in Q$ it holds that $|aq|>\lambda-1$ and $|bq|>\lambda-1$.
\end{lemma}
\begin{proof}
	Since $q \in L$, $|aq|\leqslant 1$ and $|bq|\leqslant 1$. Since $q\notin E_1$, $|aq|+|bq|>\lambda$. Combining these inequalities yields $|aq|>\lambda - 1$ and $|bq|>\lambda - 1$.
\end{proof}

The next ``helper lemma'' is the place where we use the Steiner ratio.

\begin{lemma}
	\label{Steiner-lemma-2}
	For any two points $q\in Q$ and $p\in\mathbb{R}^2$ it holds that $|pa|+|pb|+|pq|> 3\delta$. 
\end{lemma}

\begin{proof}
	Put $P'=\{a,b,q\}$, and define $L'=D(a,|ab|)\cap D(b, |ab|)$ as in Figure~\ref{noncrossing-fig1}(a). If $q\in L'\!\setminus\! E_1$ then $\textrm{Min-ST}(P')$ has edges $aq$ and $bq$, and thus $\len{\textrm{Min-ST}(P')}=|aq|+|bq|>\lambda\geqslant 6\delta/\sqrt{3}$. If $q\in L\!\setminus\!L'\!\setminus\! E_1$ then $\textrm{Min-ST}(P')$ has the edge $ab$ together with the shorter of $aq$ and $bq$, which we may assume it is $aq$ by symmetry; this case is depicted in Figure~\ref{noncrossing-fig1}(a). Thus $\len{\textrm{Min-ST}(P')}=|aq|+|ab|>\lambda-1+|ab|=6\delta/\sqrt{3}$, where the inequality is implied by Lemma~\ref{qa-qb-2}. Therefore in all cases we have $\len{\textrm{Min-ST}(P')}>6\delta/\sqrt{3}$.
	It follows from the Steiner ratio (for three points) that  
	\begin{linenomath*}
		\[|pa|+|pb|+|pq|\geqslant \len{\textrm{SMT}(P')}\geqslant \frac{\sqrt{3}}{2}\cdot\len{\textrm{Min-ST}(P')}> \frac{\sqrt{3}}{2}\cdot \frac{6\delta}{\sqrt{3}}=3\delta.\qedhere
	\]\end{linenomath*}
\end{proof}

\begin{figure}[htb]
	\centering
	\setlength{\tabcolsep}{0in}
	$\begin{tabular}{cc}
	\multicolumn{1}{m{.5\columnwidth}}{\centering\includegraphics[width=.39\columnwidth]{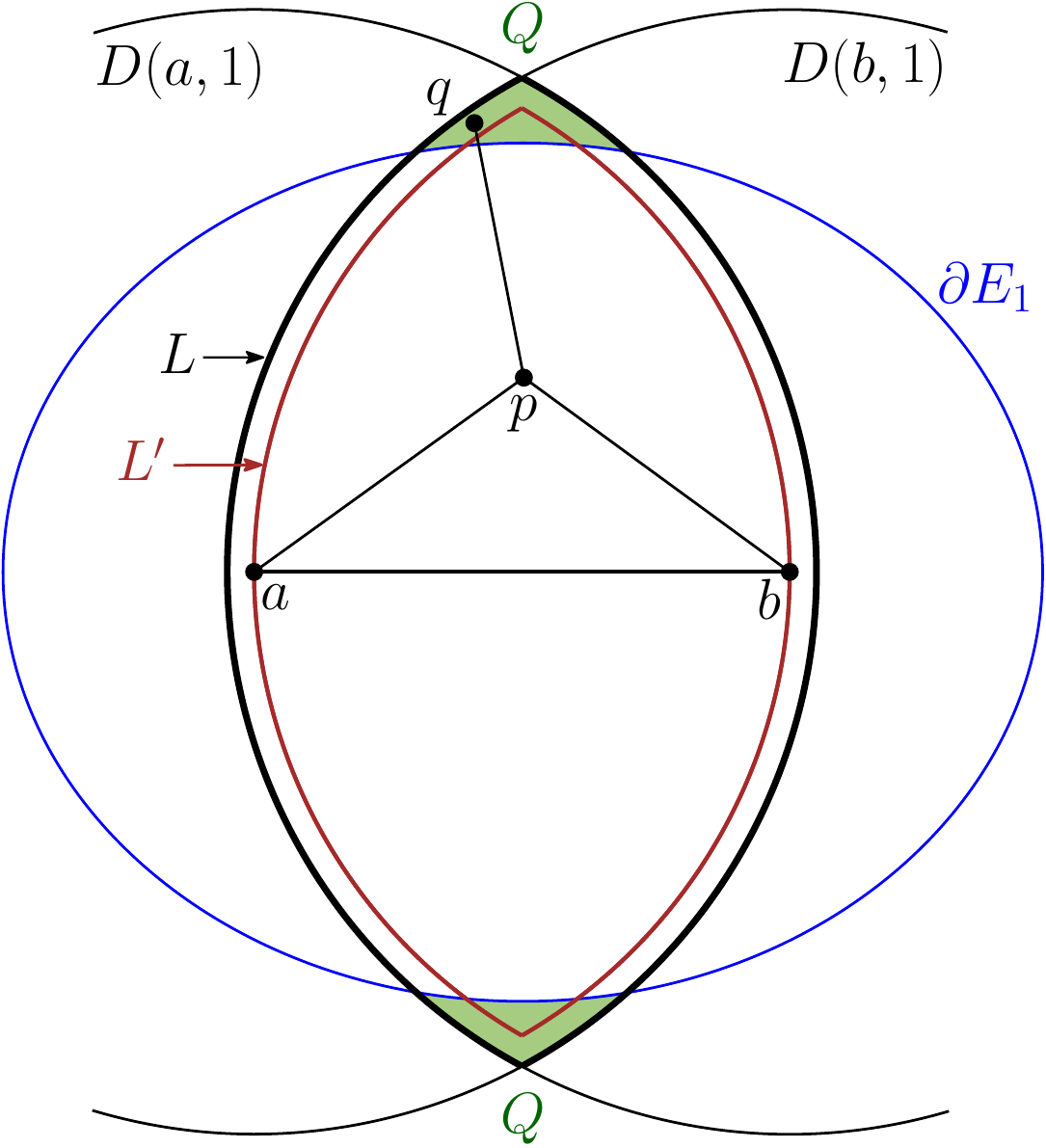}}
	&\multicolumn{1}{m{.5\columnwidth}}{\centering\vspace{0pt}\includegraphics[width=.4\columnwidth]{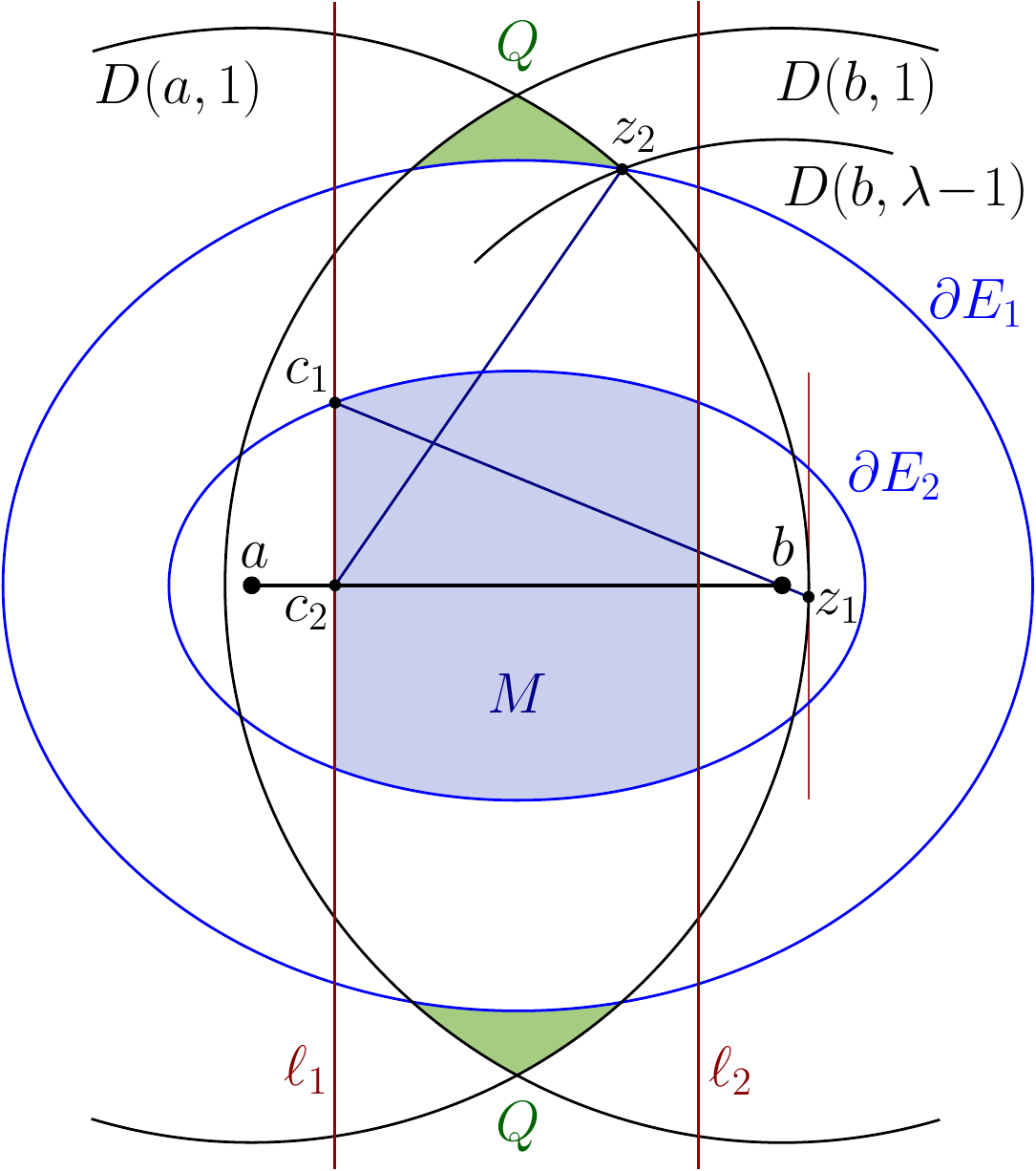}}
	\\
	(a)&(b)
	\end{tabular}$
	\caption{(a) Illustration of the case where $P\cap Q\neq\emptyset$ and $q\in L\!\setminus\!L'\!\setminus\! E_1$. (b) Illustration of the longest edges starting from $c_1$ and $c_2$ when $P\cap Q=\emptyset$.}
	\label{noncrossing-fig1}
\end{figure}

Lemmas \ref{old-helper} and \ref{beta-lemma} distinguish between two cases where $P\cap Q\neq \emptyset$ and $P\cap Q= \emptyset$. Both lemmas benefit from our helper Lemma~\ref{Steiner-lemma-2}. A combination of these three lemmas lead to a significant improvement on the approximation ratio. In Lemma~\ref{old-helper} we use the helper lemma directly to obtain a long tree for the case $P\cap Q\neq \emptyset$ which is a bottleneck case for previous algorithms. In Lemma \ref{beta-lemma} we use the helper lemme indirectly to obtain a better upper bound for the length of $T^*$. 

\begin{lemma}
	\label{old-helper}
	If $P\cap Q\neq \emptyset$ then there is a noncrossing spanning tree for $P$ of length at least $\delta\cdot\len{T^*}$.
\end{lemma}
\begin{proof}
	Consider any point $q\in P\cap Q$. By Lemma~\ref{qa-qb-2}, $|aq|\geqslant \lambda-1=6\delta/\sqrt{3}-|ab|$ and $|bq|\geqslant \lambda-1=6\delta/\sqrt{3}-|ab|$. Recall that $1/2\delta=d\leqslant|ab|\leqslant 1$ and $\delta=0.519$. Then
	\begin{linenomath*}
	\begin{equation}
	\label{Sx-eq-1}
	\min\left\{|ab|,|aq|,|bq|\right\}\geqslant \min\left\{\frac{1}{2\delta},\frac{6\delta}{\sqrt{3}}-|ab| ,\frac{6\delta}{\sqrt{3}}-|ab|\right\}\geqslant \min\left\{\frac{1}{2\delta},\frac{6\delta}{\sqrt{3}}-1 ,\frac{6\delta}{\sqrt{3}}-1\right\}>\delta.
	\end{equation}
	\end{linenomath*}
	Let $p_4,\dots,p_n$ denote the points in $P\setminus\{a,b,q\}$. It is implied by Lemma~\ref{Steiner-lemma-2} that
	\begin{linenomath*}
	\begin{equation}
	\label{Sx-eq-2}
	\sum_{i=4}^{n}|p_ia|+|p_ib|+|p_iq|>3\delta(n-3).
	\end{equation}
	\end{linenomath*}
	Denote by $x$ the point in $\{a,b,q\}$ that has the largest total distance to $p_4,\dots,p_n$, and denote by $y$ and $z$ the other two points. It is implied from \eqref{Sx-eq-2} that the total distance of $p_i$'s to $x$ is at least $\delta(n-3)$. In this setting, the star $S_x$ is a desired noncrossing tree because
	\begin{linenomath*}
	\[\len{S_x}=|xp_4|+\dots+|xp_n|+|xy|+|xz|> \delta(n-3)+\delta+\delta=\delta(n-1)\geqslant \delta\cdot\len{T^*},\]
	\end{linenomath*} where the first inequality is implied by \eqref{Sx-eq-1} and the second inequality is implied by \eqref{T-star-bound}.
\end{proof}

Subdivide $L$ into three parts by two vertical lines $\ell_1$ and $\ell_2$ at $\omega|ab|$ and $(1-\omega)|ab|$, respectively (a similar subdivision is used in \cite{Biniaz2019, Cabello2020, Dumitrescu2010}). Define $E_2=\{x\in\mathbb{R}^2: |xa|+|xb|\leqslant \gamma\}$. 
Let $M$ be the part of $L\cap E_2$ between $\ell_1$ and $\ell_2$. See Figure~\ref{noncrossing-fig1}(b). Let $\alpha$ be the fraction of points in $L\!\setminus\!E_2$, and let $\beta$ be the fraction of points in $M$. Observe that $1-(\alpha+\beta)$ fraction of points lie in parts of $L\cap E_2$ that are to the left of $\ell_1$ and to the right of $\ell_2$.

The next lemma is the place where we use the noncrossing property of $T^*$, similar to that of \cite[Lemma 3.2]{Cabello2020}.
Since Lemma \ref{old-helper} takes care of points in $Q$, we turn our attention to the case where no point of $P$ lies in $Q$. This constraint helps us to obtain a better upper bound on the length of $T^*$, which in turn leads to a better lower bound on the maximum length of the stars $S_u$ and $S_v$. Recall that $(u,v)$ is a diametral pair of $P$ and that $|uv|=1$.

\begin{lemma}
	\label{beta-lemma}
	If $P\cap Q=\emptyset$ and $\beta\geqslant\hat{\beta}$ then $\max\{\len{S_u},\len{S_v}\}\geqslant \delta\cdot\len{T^*}$.
\end{lemma}

\begin{proof}
	Recall that $ab$ is a longest edge of $T^*$. Since $P\cap Q=\emptyset$, all points of $P$ lie in $L\cap E_1$. Let $c_1$ be the top intersection point of $\ell_1$ and $\partial E_2$, and let $c_2$ be the intersection point of $\ell_1$ and $ab$, as in Figure~\ref{noncrossing-fig1}(b). If we ignore symmetry then it follows from convexity that the longest possible edge starting in $M$ and not intersecting $ab$ has either $c_1$ or $c_2$ as an endpoint. The longest edge starting from $c_1$ and not intersecting $ab$ ends at the intersection point $z_1$ of the line through $c_1$ and $b$ with the boundary of $L$. The longest edge starting from $c_2$ ends at the intersection point $z_2$ of the boundaries of $L$ and $E_1$. See Figure~\ref{noncrossing-fig1}(b) for an illustration of these edges. Observe that $z_2$ is also an intersection point of the boundaries of $D(a,1)$ and $D(b, \lambda-1)$. By using the Pythagorean theorem and the ellipse and circle equations, we can obtain the coordinates of $c_1,c_2,z_1,z_2$, and give the following expressions for $|c_1z_1|$ and $|c_2z_2|$:
	\begin{linenomath*}
	\[
	|c_2z_2|=\sqrt{\left(\frac{1+|ab|^2-\left(\lambda-1\right)^2}{2|ab|}-\omega|ab|\right)^2 + 1-\left(\frac{1+|ab|^2-\left(\lambda-1\right)^2}{2|ab|}\right)^2}
	\]
	\end{linenomath*}
		
	\begin{linenomath*}
	\[
	|c_1b|=\sqrt{\left(1-\omega\right)^2|ab|^2+\frac{\left(\gamma^2-|ab|^2\right)}{\gamma^2}\left(\left(\frac{\gamma}{2}\right)^2-\left(\frac{|ab|}{2}-\omega|ab|\right)^2\right)}
	\]
	\end{linenomath*}
		
	\begin{linenomath*}
	\[|c_1z_1|\leqslant |c_1b|+\frac{\left(1-|ab|\right)|c_1b|}{\left(1-\omega\right)|ab|}
	\]
	\end{linenomath*}
	where the upper bound on $|c_1z_1|$ is obtained by extending $c_1z_1$ to intersect the vertical line through the rightmost point of $L$; this line is depicted in Figure~\ref{noncrossing-fig1}(b).
	
	\begin{figure}[htb]
		\centering
		\includegraphics[width=.9\columnwidth]{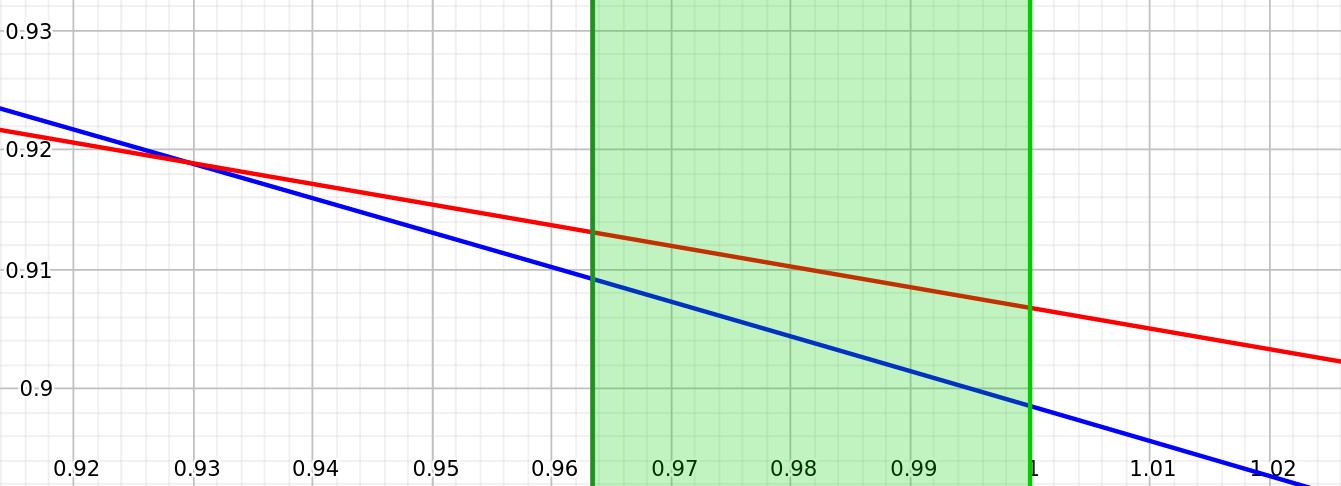}
		\caption{Plots of $|c_1z_1|$ (in red) and $|c_2z_2|$ (in blue) over $|ab|$ in interval $[d,1]$ (in green).}
		\label{plot}
	\end{figure} 
	
	Both $|c_1z_1|$ and $|c_2z_2|$ depend only on $|ab|$, and thus we denote them by functions $f_1(|ab|)$ and $f_2(|ab|)$, respectively. By considering the plots of these functions (Figure~\ref{plot}) it follows that for any $|ab|$ in interval $[d,1]$ we have $f_1(d)\geqslant f_1(|ab|)$ and $f_1(|ab|)>f_2(|ab|)$. Therefore, the longest possible edge starting in $M$ has length at most $f_1(d)$. By plugging the chosen constants in the above formula we get $f_1(d)\approx 0.913117<0.914$. Hence we obtain the following upper bound for the length of $T^*$:
	
	\begin{linenomath*}
	\[\len{T^*}\leqslant (1-\hat{\beta})n+0.914\hat{\beta}n=(1-0.086\hat{\beta})n<0.963n.\]
	\end{linenomath*}
	Recall that $\max\{\len{S_u},\len{S_v}\}\geqslant n/2$. Therefore,
	\begin{linenomath*}
	\[\frac{\max\{\len{S_u},\len{S_v}\}}{\len{T^*}}\geqslant\frac{0.5n}{0.963n}>0.519=\delta.\qedhere\]
	\end{linenomath*}
\end{proof}

The following two lemmas do not use the constraint $P\cap Q=\emptyset$.
The next lemma, presented in \cite{Cabello2020}, is adapted to work for our definition of $\hat{\alpha}$ and $\gamma$. 

\begin{lemma}
	\label{alpha-lemma}
	If $\alpha\geqslant\hat{\alpha}$ then $\max\{\len{S_a},\len{S_b}\}\geqslant \delta\cdot\len{T^*}$.
\end{lemma}
\begin{proof}
The proof is somewhat similar to that of Lemma~\ref{two-star-lemma}, except now we have a better lower bound for points in $L\!\setminus \! E_2$, which is $\gamma$. Thus
\begin{linenomath*}
\[\max\{\len{S_a},\len{S_b}\}\geqslant \frac{\len{S_a}+\len{S_b}}{2}\geqslant \frac{\alpha n\cdot\gamma+(1-\alpha)n\cdot|ab|}{2}=\frac{n(|ab|+\alpha(\gamma-|ab|))}{2}.\]
\end{linenomath*}
This and the fact that $\len{T^*}<n|ab|$ by \eqref{T-star-bound}, imply that
\begin{linenomath*}
\[\frac{\max\{\len{S_a},\len{S_b}\}}{\len{T^*}}> \frac{n(|ab|+\alpha(\gamma-|ab|))}{2n|ab|}\geqslant \frac{|ab|+\hat{\alpha}(\gamma-|ab|)}{2|ab|}=\delta.\qedhere\]
\end{linenomath*}
\end{proof}


The construction in the following lemma adapted from one by Biniaz \etal~\cite{Biniaz2019}. We refine the construction according to our choice of points $a$, $b$, and refine the analysis according to our definition of $\omega$, $\hat{\alpha}$, and $\hat{\beta}$. 

\begin{lemma}
	\label{alpha-beta-lemma}
	If $\alpha\leqslant \hat{\alpha}$ and $\beta\leqslant \hat{\beta}$ then there is a noncrossing spanning tree for $P$ of length at least $\delta\cdot\len{T^*}$. 
\end{lemma}

\begin{proof}
We keep the description short. We construct two trees $T_a$ and $T_b$ such that the longer one is a desired tree. We describe the construction for $T_a$; the construction of $T_b$ is analogous.
\let\qed\relax\end{proof}
\vspace{-8.5pt}
\begin{wrapfigure}{r}{1.75in} 
\centering
\vspace{+1pt} 
\includegraphics[width=1.6in]{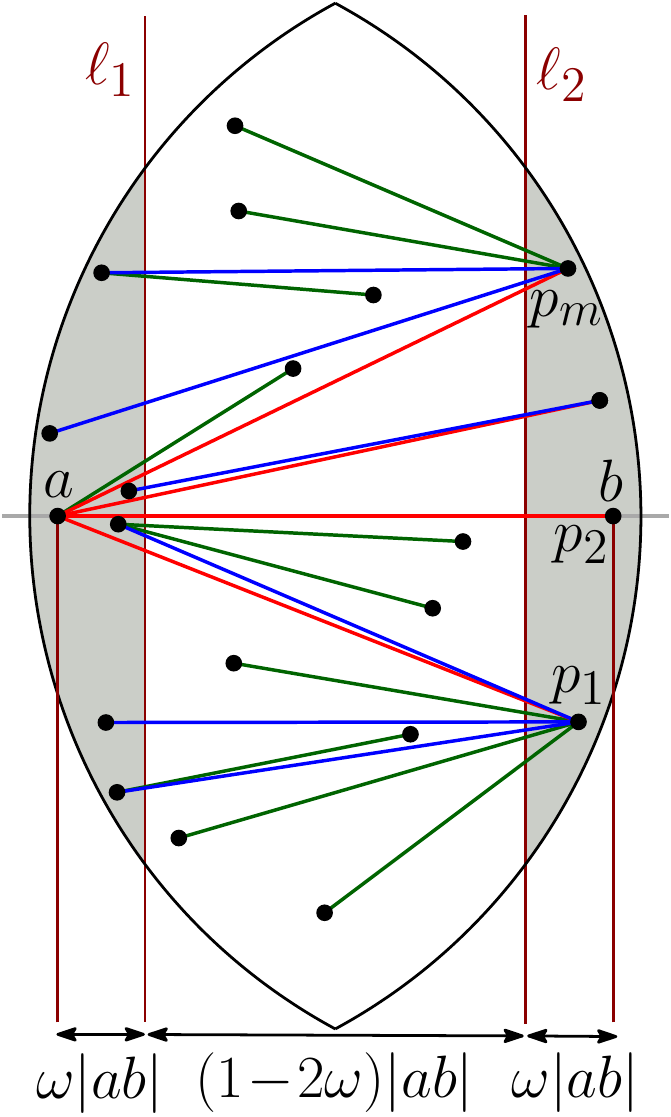} 
\vspace{-5pt} 
\end{wrapfigure}\noindent 
\indent See the figure to the right for an illustration. Start by connecting all points to the right of $\ell_2$ to $a$ (red edges). Each such edge has length at least $(1-\omega)|ab|$. Let $ap_1,\dots, ap_m$ be these edges in radial order around $a$ as in the figure; notice that $m\geqslant 1$ because $b$ is to the right of $\ell_2$. Now we connect the points to the left of $\ell_1$: connect all points below $ap_1$ (that are not above $ab$) to $p_1$, connect all points above $ap_m$ (that are not below $ab$) to $p_m$, and connect all points between two consecutive edges $ap_i$ and $ap_{i+1}$ to $p_i$ for $1\leqslant i\leqslant m-1$ (blue edges). Each new edge has length at least $(1-2\omega)|ab|$. 
	
Finally we connect the points in the region between $\ell_1$ and $\ell_2$. Let $\beta'$ be the fraction of points in this region. This region is subdivided into subregions  by the current (red and blue) edges of $T_a$. Each subregion is bounded by one or two edges of $T_a$ in such a way that at least one edge is fully visible from the interior of the subregion. Connect all points in each subregion to the endpoint (of the visible edge) that gives a larger total distance; these new edges are shown in green. By an argument similar to the proof of Lemma~\ref{two-star-lemma} one can show that the total length of the new edges is at least $\beta'n(1-2\omega)|ab|/2$. 
 
By the definition of $\alpha$, $\beta$, and $\beta'$ it holds that $\beta'\leqslant \alpha+\beta$. Since $\alpha\leqslant \hat{\alpha}$ and $\beta\leqslant \hat{\beta}$, we have $\beta'\leqslant \hat{\alpha}+\hat{\beta}$. The total fraction of points to the left of $\ell_1$ and to the right of $\ell_2$ is $1-\beta'$. Thus
\begin{linenomath*}
\begin{align}
\notag \len{T_a}+\len{T_b}&\geqslant (1-\beta')n(2-3\omega)|ab|+\beta'n(1-2\omega)|ab|\\ \notag
&=(2-3\omega+(\omega-1)\beta')n|ab|\\ \notag &\geqslant (2-3\omega+(\omega-1)(\hat{\beta}+\hat{\alpha}))n|ab|,
\end{align}
\end{linenomath*} 
where the last inequality holds because $\omega-1<0$. Therefore 
\begin{linenomath*}
\[
\pushQED{\qed} 
\max\left\{\len{T_a}+\len{T_b}\right\}\geqslant\frac{2-3\omega+(\omega-1)(\hat{\beta}+\hat{\alpha})}{2}\cdot n|ab|=\delta \cdot n|ab|> \delta\cdot\len{T^*}.\qedhere
\popQED\]\end{linenomath*}

The cases considered in Lemmas \ref{ab-small}, \ref{old-helper}, \ref{beta-lemma}, \ref{alpha-lemma}, and \ref{alpha-beta-lemma}  ensure that at least one of $S_u$, $S_v$, $S_a$, $S_b$, $T_a$, $T_b$, and $S_q$ (introduced in Lemma~\ref{old-helper}) has length at least $\delta\cdot\len{T^*}$. This concludes our analysis and proof of Theorem~\ref{noncrossing-thr}.

\section{Conclusions}

A natural open problem is to improve the presented bounds ($0.524$ for the Max-ST-NB problem and $0.519$ for the Max-NC-ST problem) further by employing more new ideas.
We believe there is a possibility to slightly improve both bounds by discretization. For example in our Max-ST-NB (resp. Max-NC-ST) algorithm if the points in the lens $L'$ (resp. the region $M$) are close to $ab$ then one could obtain a better upper bound on the length of $T^*$ otherwise obtain a better lower bound on the length of the double-star $D$ (resp. the star $S_a$ or $S_b$).
However, the improvement would be minor, and the required case analysis may hide the impact and beauty of the main techniques. 

The improved approximation ratios are obtained mainly by employing the Steiner ratio, which has not been used in this context earlier. It would be interesting to see if the Steiner ratio can be used to design better algorithms for other related problems.


\bibliographystyle{abbrv}
\bibliography{Long-Tree-Neighborhood}
\end{document}